\documentclass[12pt]{article}
\usepackage{graphicx}
\usepackage{amsmath,amsthm,amsfonts}
\usepackage{courier}
\usepackage{amssymb}
\usepackage{cite}

\newtheorem{theorem}{Theorem}[section]
\newtheorem{proposition}[theorem]{Proposition}

\newtheorem{lemma}[theorem]{Lemma}

\theoremstyle{definition}

\newtheorem{def-theorem}[theorem]{Definition-Theorem}
\newtheorem{remark}[theorem]{Remark}
\newtheorem{definition}[theorem]{Definition}
\newtheorem{example}[theorem]{Example}

\newcommand{\be}{\begin{equation}}
\newcommand{\ee}{\end{equation}}
\newcommand{\bea}{\begin{eqnarray}}
\newcommand{\eea}{\end{eqnarray}}
\newcommand{\beas}{\begin{eqnarray*}}
\newcommand{\eeas}{\end{eqnarray*}}
\newcommand{\ba}{\begin{array}}
\newcommand{\ea}{\end{array}}

\newcommand{\tr}{{\rm tr}}

\begin{document}

\begin{titlepage}
\hfill
\vbox{
    \halign{#\hfil         \cr
           } 
      }  
\vspace*{10mm}
\begin{center}
{\Large \bf Finite entropy sums in quantum field theory}

\vspace*{5mm}
\vspace*{1mm}
 Mark Van Raamsdonk
\vspace*{1cm}
\let\thefootnote\relax\footnote{mav@phas.ubc.ca}

{Department of Physics and Astronomy,
University of British Columbia\\
6224 Agricultural Road,
Vancouver, B.C., V6T 1Z1, Canada
}

\end{center}
\begin{abstract}
Entropies associated with spatial subsystems in conventional local quantum field theories are typically divergent when the spatial regions have boundaries. However, in certain linear combinations of the entropies for various subsystems, these divergences may cancel, giving finite quantities that provide information-theoretic data about the underlying state. In this note, we show that all such quantities can be written as linear combinations of three basic types of quantities: i) the entropy of a spatial subsystem minus the entropy of its complementary subsystem, ii) the mutual information between non-adjacent subsystems, and iii) the tripartite information for triples of disjoint subsystems. For a fixed decomposition of a spatial slice into regions, we describe a basis of sums of entropies for collections of for these regions for which all divergences related to both region boundaries and higher-codimension intersections of regions cancel. Key mathematical technology used in this work (Fourier transforms on the Boolean cube and Möbius transformations of functions on partially ordered sets) and several of the main proof ideas were suggested by AI (ChatGPT5). We offer a few comments on the use of AI in physics and mathematics, based on our experience. 
\vspace*{1cm}

{\it Dedicated to my father, Ray Van Raamsdonk 1945-2025, who introduced me to the beauty of mathematics and who was one of AI's biggest fans since the 1980s.}

\end{abstract}
\vspace*{5mm}

\end{titlepage}

\section{Introduction}

\begin{figure}
\label{fig:ThreeTypes}
    \centering
\includegraphics[width=0.7\textwidth]{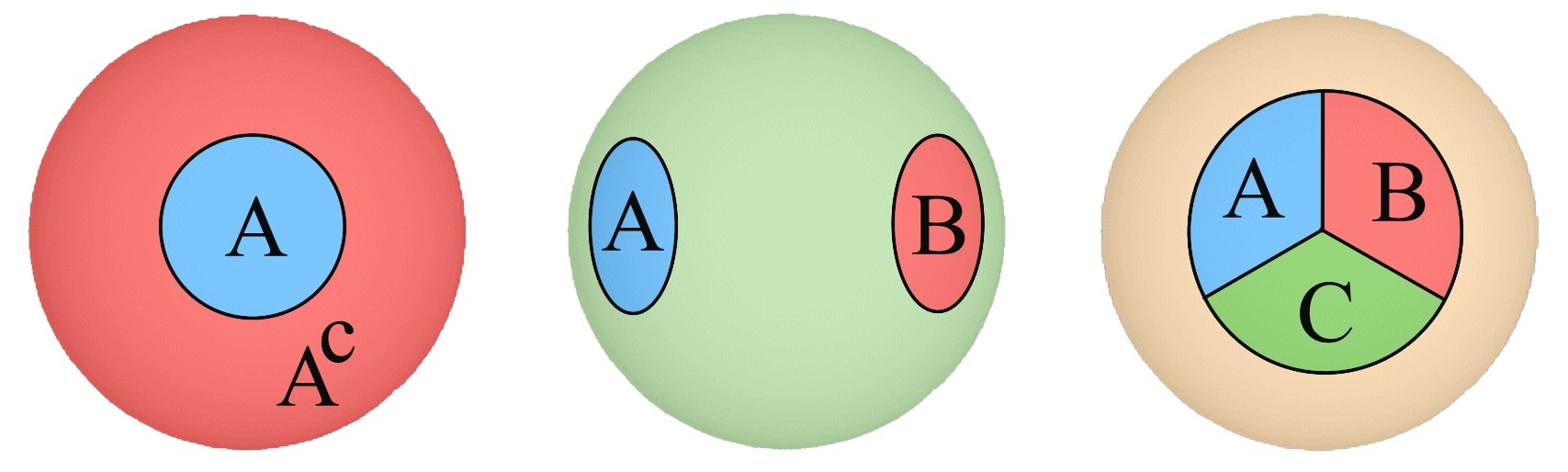}
    \caption{Example regions for the three basic types of finite entropy combinations: i) $S(A) - S(A^c)$, the entropy of a region minus the entropy of its complement ii) $I(A,B)$, the mutual information between a pair of non-adjacent regions, iii) $I(A,B,C)$ the tripartite information for a triple of disjoint regions.}
\end{figure} 

In general quantum systems, subsystem entropies\footnote{If $\rho_A$ is the reduced density operator associated with a subsystem $A$, the subsystem entropy, also known as entanglement entropy is defined by the von Neumann entropy $S_A = - \tr(\rho_A \log \rho_A)$ \cite{Wehrl1978}.} provide insight into the structure of entanglement and correlation present in the global quantum state. For local quantum field theories, such subsystem entropies are generally divergent for subsystems describing the degrees of freedom in a spatial region with a non-trivial boundary (see e.g. \cite{Bombelli1986,Srednicki1993,CasiniHuerta2009}). Physically, this divergence indicates an infinite amount of entanglement between quantum field theory UV degrees of freedom on either side of the boundary.\footnote{Additional divergences can arise in the case of infinite volume, but for this paper, we will always assume that the quantum field theory is defined on a manifold with compact spatial slices.} \footnote{More formally, the non-existence of finite subsystem entropies for regions with boundaries is related to the fact that the associated subalgebras of observables are type III; see \cite{Haag1996,BratteliRobinson1997,Witten2018} for background.}

While these individual entropies are divergent, by considering certain linear combinations of entropies associated with various regions, it is possible to define finite quantities that provide information about the quantum field theory state.\footnote{In practice, these are defined by introducing a regulator that renders the individual entropies finite, and then noting that in a limit where the regulator is removed, the linear combinations in question yield finite results that are independent of the regularization scheme.} These quantities are finite either because they measure the entropies of subsystems without boundaries or because they involve combinations of entropies whose divergences cancel between terms. For a linear combination of entropies 
\begin{equation}
\label{entsum}
\sum_i c_i S(A_i)
\end{equation}
with coefficient $c_i$ for the entropy of region $A_i$, cancellation requires that any part of the boundary of some region $A_j$ in the sum also appears as part of the boundary of other regions in the sum, in such a way that the coefficients for all such regions add to zero. Roughly, this cancellation works because the parts of the regulated entropies that diverge when the regulator is removed are local functionals of the boundary geometry for the region.

Additional divergences may arise from higher-codimension intersections of $m>2$ regions \cite{CasiniHuerta2007_Universal2p1,CasiniHuerta2009_Review,BuenoMyersWitczak2015_PRL,BuenoMyers2015_JHEP,BuenoMyersWitczak2015_Twist,KallinEtAl2014_JStatMech,StoudenmireEtAl2014_PRB,HaywardSierensEtAl2017_TrihedralPRB,BednikEtAl2019_TrihedralPRB,SeminaraSistiTonni2017_BCFT,Berthiere2019_BoundaryCornerPRB,Nishioka2018_RMP}. In order to cancel these, we require that for any region corner appearing at a multi-region intersection, the the entropy coefficients for all regions containing either this corner or the complementary corner add to zero.\footnote{Here, we only need to worry about region corners with an associated divergence. In some cases for sufficiently cuspy corners such a divergence may not appear \cite{BuenoCasiniWitczak2019_SingularJHEP}.}

In this paper, we consider a general situation where the spatial slice on which a quantum field theory state is defined is divided into $n$ regions, and study entropy sums of the form (\ref{entsum}) with terms corresponding to unions of various subsets of these regions. Three basic types of entropy sums for which codimension-one boundary divergences cancel are
\begin{enumerate}
    \item The entropy of a subsystem minus the entropy of its complement, $S(A) - S(A^c)$. 
    \item The mutual information $I(A,B) \equiv S(A) + S(B) - S(A \cup B)$ between two regions that do not share a codimension-one boundary component.
    \item The tripartite information  $I(A,B,C) = S(A) + S(B) + S(C) - S(A \cup B) - S(A \cup C) - S(B \cup C) + S(A \cup B \cup C)$ for  three disjoint regions.
\end{enumerate}
Simple examples of regions for which each of these three types apply are shown in Figure \ref{fig:ThreeTypes}. In this paper, we
\begin{itemize}
    \item Show that all entropy sums with cancellation of codimension-one boundary divergences are linear combinations of these three basic types of quantities. 
    \item Describe two different explicit bases for the space of  entropy sums with cancelling codimension one divergences. For a given set of $n$ regions with $B$ pairs having a codimension-one intersection, the dimension of this space is $2^n - B -1$.
    \item
    Describe an explicit basis for the smaller space of entropy sums which additionally cancel all divergences associated with higher codimension features. The dimension of this space is $2^n - |I_E| - 1$, where $|I_E|$ is the number of even order subsets of regions meeting at some intersection. 
\end{itemize} 

When taking into account divergences associated with higher codimension intersections, we show that finiteness of a mutual information generally requires that the two regions are non-adjacent i.e. do not meet even at a higher codimension intersection. Finiteness of a tripartite information generally requires that there is no intersection where the three participating regions intersect with a fourth region. We note however, that in special cases, mutual informations or tripartite informations that would generically be divergent based on their intersection properties can be finite if region corners participating in the intersection is sufficiently cuspy \cite{CasiniHuerta2007_Universal2p1} (see Figure \ref{fig:Examples}).

The type 1 entropy combination can be related to mutual information if we introduce a purifying system $P$ such that the state of the quantum field theory together with the purifying system is pure. In this case, we have
\[
S(A) - S(A^c) = [S(A^c \cup P) - S(A^c) - S(P)] + S(P) = I(A^c,P) + S(\Sigma) \; .
\]

The tripartite information can also formally be written in terms of mutual information, for example via
\[
I(A,B,C) = I(A,B) + I(A,C) - I(A , B \cup C) 
\]
but the individual mutual information terms here are only finite if A not adjacent to either B or C. Thus, for triples of regions where each is adjacent to at least one of the others, the tripartite information should be considered as a basic quantity. The tripartite information  $I(A,B,C)$ makes an important physical appearance in the definition of topological entanglement entropy \cite{KitaevPreskill2006,LevinWen2006} and has been shown to be positive in holographic quantum field theories \cite{hayden2013holographic}.

\paragraph{Proof strategy} In the remainder of the introduction, we outline the basic steps that we use to prove our main results. In the following, we begin by considering only divergences associated with codimension-one intersections of regions, then generalize to the higher codimension case.
\begin{enumerate}
\item We begin by noting that for any spatial geometry $\Sigma$ divided into regions $A_i$, we can represent the relevant information about the number of regions and their codimension-one intersections though a graph $G$ with one vertex for each region and an edge between vertices $i$ and $j$ corresponding to regions that are adjacent.
\item Next, we note that linear combinations of entanglement entropies for various regions are in one-to-one correspondence with functions that assign a real number to each subset of vertices; the value of the function gives the coefficient of the entropy for region formed by the union of regions associated with the vertices in the subset. This set of functions forms a real vector space ${\cal T}_\mathbb{R}$ with a natural inner product $(T_1,T_2) = \sum_\sigma T_1(\sigma) T_2(\sigma)$.
\item 
The set of functions we are interested is a subspace ${\cal T}_G$ of ${\cal T}_\mathbb{R}$. We show that ${\cal T}_G$ is the subspace orthogonal to particular set of functions $g_e$ associated with the edges of the graph, or alternatively the kernel of a certain operator ${\cal E}$ mapping ${\cal T}_\mathbb{R}$ to the vector space $\mathbb{R}^{|E|}$ of real functions on the set of edges. For $T \in {\cal T}_\mathbb{R}$, the component of ${\cal E}(T)$ corresponding to an edge $e$ is the sum of all coefficients $T(\sigma)$ where $e$ is in the cut of $\sigma$ (i.e. has one vertex in $\sigma$ and one vertex not in $\sigma$). This corresponds to the net number of times a certain boundary component appears in the sum over entropies associated with $T$. We need this to vanish in order that the divergences can cancel.
\item 
In order to characterize the kernel of ${\cal E}$, we recall that there is a natural analog of the Fourier transform for functions on subsets of a set (sometimes known as functions on the {\it Boolean cube}). Any function in ${\cal T}_\mathbb{R}$ can be expanded in a Fourier basis that is also labeled by subsets of the set (in our case, the set of vertices of $G$). We show that a function is in the kernel of ${\cal E}$ if and only if the Fourier coefficients for basis elements corresponding to edges of the graph are all equal to the Fourier coefficient corresponding to the empty set. Based on this, we give simple orthogonal bases for ${\cal T}_G$ and for its orthogonal complement ${\cal T}_G^\perp$ constructed from the Fourier basis elements (Theorem \ref{thm:Fourier}). The size of this basis, and thus the dimension of ${\cal T}_G$ is $2^{|V|} - |E|$.\footnote{This is one larger than the dimenion of the space of finite entropy sums since it includes a function that assigns 1 to the empty set and 0 to all other subsets.}
\item 
Functions in this Fourier-based orthogonal basis are somewhat unwieldy since they have non-zero coefficients for every subset. We show that an alternative basis can be chosen from functions corresponding to the three types of entropy combinations described above, plus a function that assigns 1 to the empty set and 0 to all other subsets (Theorem \ref{thm:Möbius}). First, defining ${\cal T}_G^*$ to be the span of these four special types of functions, we verify that ${\cal T}_G^* \subset {\cal T}_G$.\footnote{Throughout this paper $\subset$ is taken to mean any subset, proper or improper. We also occasionally use $\subseteq$ to mean the same thing.}  We complete our demonstration by showing that $({\cal T}_G^*)^\perp \subset {\cal T}_G^\perp$. To do this, we make use of one more mathematical tool, the idea of a Möbius transform for functions on a partially ordered set (in this case, the set of subsets of vertices). Any function in ${\cal T}_\mathbb{R}$ can be expanded in terms of a Möbius basis that is again labeled by subsets of a vertices. We show that orthogonality to the various types of functions in ${\cal T}_G^*$ implies simple constraints on the Möbius coefficients. The full constraint $T \in ({\cal T}_G^*)^\perp$ implies $T$ has a very simple form in terms of Möbius basis elements associated with vertices or edges. It is easy to show that all such functions are in ${\cal T}_G^\perp$. 
\item 
Finally, we generalize to consider higher-codimension intersections of three or more regions. These may be captured by promoting our graph to a {\it hypergraph} $H$, which has, in addition to vertices and edges, sets of $m$-vertex collections $H_m$ ($m > 2$) representing the higher-order intersections. We show that our divergence-cancellation condition for these intersections may also be expressed very simply in terms of the Fourier coefficients, and leads to the additional restriction on mutual information and tripartite information noted above. The primary results are expressed in Theorem \ref{thm:Fourier2} and Proposition \ref{thm:Möbius2}.
\end{enumerate}

Examples of various simple sets of regions, together with their corresponding graphs and a list of independent finite entropy combinations are shown in figure \ref{fig:graphs}. The various basis function  for functions on the Boolean cube corresponding to a set of 3 elements are displayed in Figure \ref{fig:Boolean1}.

\begin{figure}
\label{fig:graphs}
    \centering
\includegraphics[width=0.9\textwidth]{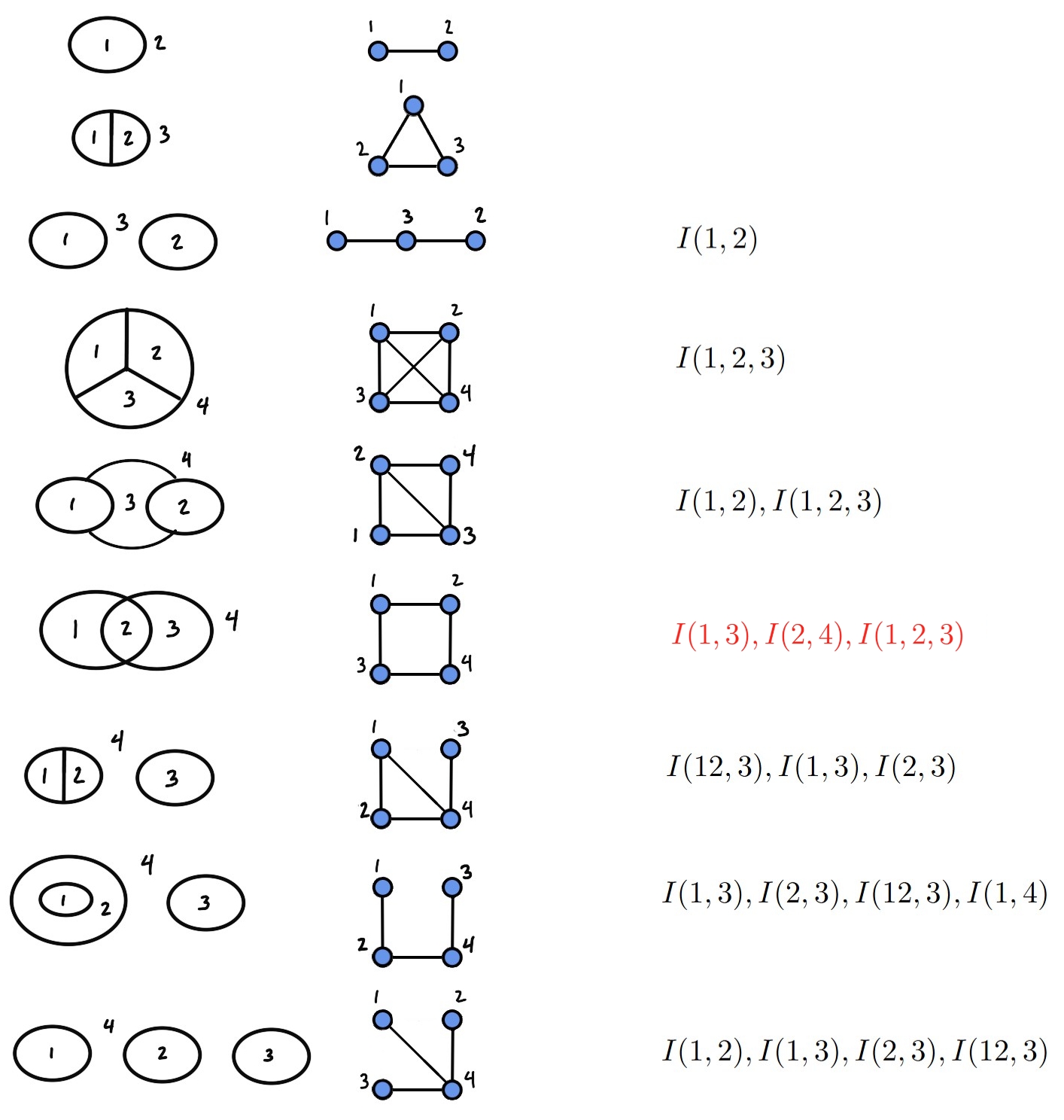}
    \caption{Example regions, associated graphs, and bases of entropy sums for which divergences associated with of codimension-one boundaries cancel. Here, $I(A,B) = S(A) + S(B) - S(A,B)$ and $S(A,B,C) = S(ABC) - S(AB) - S(BC) - S(AC) + S(A) + S(B) + S(C)$. In addition to the quantities shown we have also the quantity $S(R) - S(R^c)$ for any region $R$. Quantities in red have divergences associated with higher-order intersections and must be excluded to form a basis of fully finite entropy sums.}
\end{figure} 

Here is an outline of the rest of the paper. In section 2, we provide the basic setup and reformulate the problem of finding entropy sums with cancelling divergences as a problem of characterizing a certain map on the space functions assigning real numbers to subsets of vertices of a graph. In section 3, we state and prove our main results related to the cancellation of divergences associated with codimension-one boundaries, describing two different ways to construct a basis for the space of sums where these divergences cancel. In section 3, we take into account the possibility of additional divergences associated with higher-codimension intersections of multiple regions, showing that cancellation of these divergences gives additional constraints on Fourier coefficients that are equivalent to additional restrictions on the regions that may be appear in mutual and tripartite informations. In section 5, we note that, allowing for complex coefficients in the entropy sums, there is a one-to-one correspondence between finite entropy sums for $n$ regions and ground state wavefunctions for a certain Hamiltonian on $n$ qubits. Thus, our results can also be understood as characterizing the space of ground states for this quantum mechanics problem. We end in section 6 with a brief discussion of how the divergence-free entropy sums considered in this paper may be defined algebraically without reference to divergent entropies.

We emphasize that the new results in this paper (in particular, Theorems \ref{thm:Fourier}, \ref{thm:Möbius}, \ref{thm:Fourier2}) are not quantum field theory statements but rather rigorous statements about the space of functions assigning numbers to subregions of a manifold or to subsets of vertices of a graph / hypergraph. We precisely characterize all entropy sums satisfying conditions that are widely believed to be necessary and sufficient for cancellation of divergences in a wide class of quantum field theories, but we make no attempt to 
prove that the conditions we start with are necessary and/or sufficient for divergence cancellation in any particular quantum field theory. 

\paragraph{A note about the use of AI} This paper represents the first instance for the author where the use of AI tools was an essential component of the work. A computer analysis (coded by Google Gemini 2.5) analyzing all graphs up to 7 vertices and verifying that the functions in ${\cal T}_G^*$ span all of ${\cal T}_G$ in each case provided initial strong evidence for the results of section 3. A prompt to Chat GPT5-Thinking giving the statement of Theorem \ref{thm:Möbius} as a conjecture (in graph theory language) and requesting a proof produced a proof sketch that contained essentially all the main ideas of the final proof presented in section 3, including the statement and proof sketch of Theorem \ref{thm:Fourier}. The content in section 4 was suggested after a prompt asking for suggestions of natural extensions of the work. Here, after supplying the cancellation conditions in Definition 5.1, GPT5 suggested both the main results in Theorems \ref{thm:Fourier2} and \ref{thm:Möbius2} and the basic structure of the proofs. As an example, the transcript of the conversation leading to section 3 may be found here \cite{GPT5}.

In all cases, the line-by-line proof details presented here were constructed by the author. It seems important to point out that GPT5 was {\it not} reliable in providing proof details. In several cases during the present project, prompting of GPT5 to produce some detailed part of a proof gave results that were sloppy or incorrect. In one situation, the model supplied two alternative ``proofs'' for a conjecture that turned up to be false.
While AI models are certainly capable of producing a correct proof in many cases, they also appear to excel at making incomplete proofs sound convincing or producing the most convincing possible argument for a false statement. Thus, the author recommends {\it extreme caution} when evaluating the details of an argument/proof provided by AI and suggests fully reconstructing the details in any consequential situation.

At this point, the author would heartily endorse AI as a valuable resource to suggest relevant mathematical tools and proof ideas, to carry out numerical checks, to check for typos or errors in an argument, and to suggest related previous work or potential extensions of a project. On the other hand, the author cautions that trusting the details of an AI proof without independent expert verification is akin to dancing with the devil.

\section{Graph theory formulation}

Let $\Sigma$ be a spatial slice of some spacetime geometry on which a quantum field theory is defined. Consider $n$ non-intersecting open regions $A_i$ such that the union of their closures is $\Sigma$. For a subset $\sigma \subset \{1,\dots, n\}$, we let $A_\sigma$ be the region formed by the interior of the closure of $\cup_{i \in \sigma} A_i$ and $S(\sigma)$ represent the (regulated) entropy for region $A_\sigma$. For a given collection of regions, we would like to characterize the space of linear combinations of the form 
\[
\sum_{\sigma} c_\sigma S_\sigma
\]
for which all divergences associated with the codimension-one boundary regions cancel. We will return to divergences associated with higher-codimension intersections in section 4. The cancellation of divergences requires that for any boundary component $B_{ij}$ defined by the intersection of boundaries of elementary regions $A_i$, the sum of coefficients for all regions $A_\sigma$ with $B_{ij}$ in their boundary vanishes.\footnote{We expect that these conditions are necessary and sufficient for cancellation provided that there exists a regularization scheme where i) the divergences can be cancelled by adding a state-independent regulator-dependent term that depends only on the (intrinsic or extrinsic) boundary geometry and ii) the regulated entropy of a region is equal to the entropy of the complementary region for a pure state.} That is, we require
\begin{equation}
\label{entcond}
\sum_{\sigma | B_{ij} \in \partial A_{\sigma}} c_\sigma = 0 \qquad \forall B_{ij} \; .
\end{equation}
We will show that all sums above with coefficients satisfying this condition can be written as a linear combination of the four types of quantities described in the introduction, namely the entropies of full connected components, entropy differences between regions and their complements, mutual informations between non-adjacent regions, and tripartite informations for triples of disjoint regions.  

\paragraph{Graphs} For a given set of subsystems $A_i \subset \Sigma$, the only geometrical property that comes into our conditions (\ref{entcond}) on coefficients in the entropy sum are the number of regions and their adjacency properties. This information can be conveniently captured using a graph  (see, e.g., \cite{Diestel2017}). For a set of $n$ regions subsystems $\{A_i\}$ of $\Sigma$, we define an associated graph $G$ on $n$ vertices with an edge between vertices $i$ and $j$ if and only if the regions $i$ and $j$ are adjacent (share a codimension one boundary component). An example is shown in Figure \ref{fig:GtoG}.

\begin{figure}
\label{fig:GtoG}
    \centering
\includegraphics[width=0.9\textwidth]{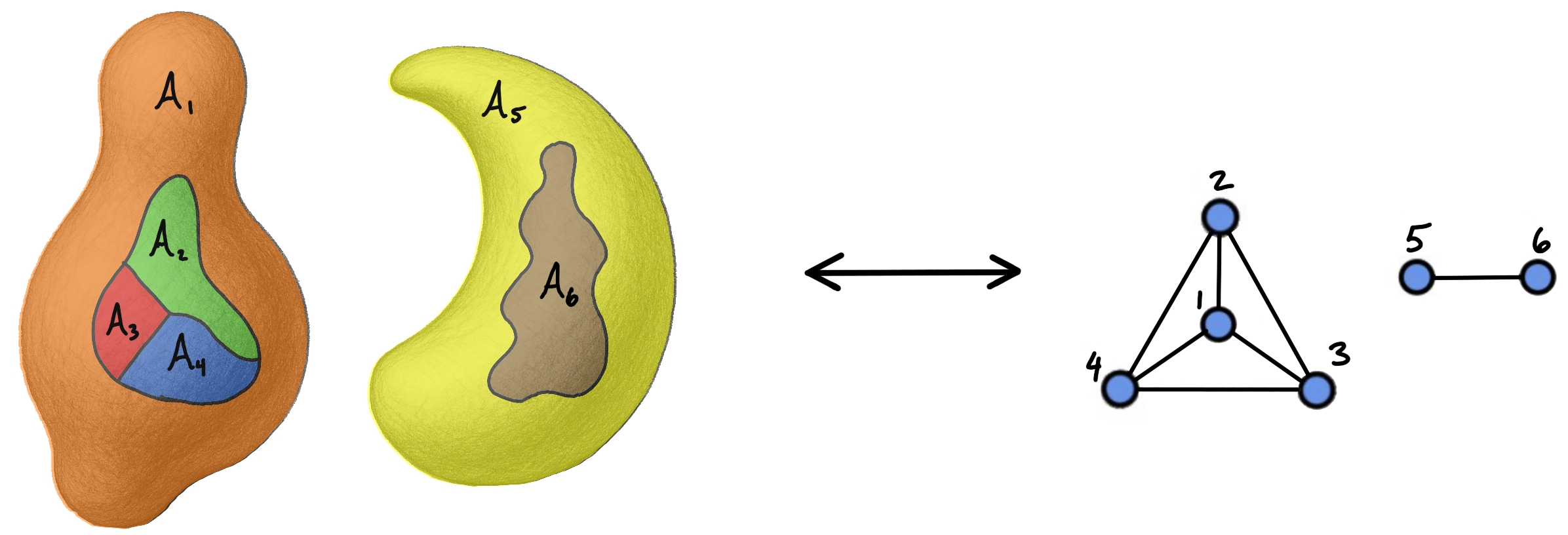}
    \caption{Example of a (disconnected) spatial manifold $\Sigma$ divided into regions $A_i$, with the associated graph $G$ capturing the adjacency properties of the regions.}
\end{figure} 

We can associate a general linear combinations of entropies with maps $T:2^V \to \mathbb{R}$ from subsets of vertices of $G$ to real numbers by
\[
T \leftrightarrow \sum_\sigma T(\sigma) S(\cup_{i \in \sigma} A_i) \; .
\]
Here, we could demand that the value of $T$ on the empty set, corresponding to the entropy of the empty region is 0, but we will mostly drop this restriction and consider the value as some arbitrary constant that can be added to the entropy sum.

The space of all such functions on subsets of vertices is a real vector space that we denote by ${\cal T}_\mathbb{R}$. The set of functions corresponding to entropy sums with cancelling divergences is a subspace ${\cal T}_G \subset {\cal T}_\mathbb{R}$ that we would like to characterize. 

It will be convenient to define an inner product on ${\cal T}_\mathbb{R}$ by
\begin{equation}
\label{inner}
(T_1,T_2) = \sum_{\sigma} T_1(\sigma) T_2(\sigma)
\end{equation}
With this inner product, there is a simple orthonormal basis for ${\cal T}_\mathbb{R}$:
\begin{definition} (The indicator basis)
Define $[\sigma]$ to be the function that assigns $1$ to the subset $\sigma$ and $0$ to all other subsets. 
\label{def:indicator}
\end{definition}
\begin{proposition}
The set of functions $\{[\sigma]\}$ for all subsets $\sigma$ of vertices form an orthonormal basis of ${\cal T}_\mathbb{R}$.
\end{proposition}
\begin{proof}
   For any $T$, we can write
\[
T = \sum_\sigma T(\sigma) [\sigma] \; .
\]
The orthonormality $([\sigma],[\tau]) = \delta_{\sigma,\tau}$ is immediate from the definition of the inner product.
\end{proof}
We can describe the subspace ${\cal T}_G \subset {\cal T}_\mathbb{R}$ of functions corresponding to entropy sums with cancelling divergences as the subspace orthogonal to a set of vectors. To define these, we introduce some standard graph theory technology: 
\begin{definition}
There is a natural map $\delta: 2^V \to 2^E$ known as the {\it cut map}. For a subset $\sigma$ of vertices, $\delta \sigma$ is the subset of edges with exactly one endpoint in $\sigma$. 
\end{definition}
The cut map is relevant to our problem since an edge $e = \{ij\}$ is in the cut of some subset $\sigma$ if and only if $A_\sigma$ if the  boundary region $B_e$ is part of the boundary of $A_\sigma$. We now have
\begin{proposition}
\label{inner}
    For each edge $e$, define a function $g_e \in {\cal T}_\mathbb{R}$ by 
    \[
    g_e(\sigma) = \left\{ \begin{array}{ll} 1 & \qquad e \in \delta \sigma \cr
    0 & \qquad {\rm otherwise} \end{array} \right.
    \]
    where $\delta$ is the cut map. Then ${\cal T}_G$ is the subspace of ${\cal T}_\mathbb{R}$ defined by $(g_e,T) = 0$ for all $e$.
\end{proposition}
\begin{proof}
   We have that 
   \[
   (g_e, T) = \sum_{\sigma | e \in \delta \sigma} T(\sigma)
   \]
   Since an edge $e$ being in the cut of $\sigma$ corresponds to the boundary component $B_e$ being in the boundary of $A_\sigma$, this exactly gives the  condition (\ref{entcond}).
\end{proof}
An alternative characterization of ${\cal T}_G$ that will be useful to us is as the kernel of a certain operator.
\begin{definition}
Defining $2^S$ as the set of subsets of a set $S$ of $n$ elements (that we will take to be $V$ or $E$), it will be convenient to define the corresponding {\it characteristic vector} $x_\sigma$ for a subset $\sigma$ to be the element of $\mathbb{R}^n$ whose $i$th component is 1 or 0 depending on whether or not the $i$th element of $S$ is in $\sigma$.
\end{definition}

\begin{definition}
Define the linear map ${\cal E}:{\cal T}_{\mathbb{R}} \to \mathbb{R}^{|E|}$ by its action on the indicator basis $\{[\sigma]\}$ via
\[
{\cal E} \cdot [\sigma] = x_{\delta \sigma} \; .
\]
where $\delta$ is the cut map and $x_{\delta \sigma}$ is the characteristic vector for $\delta \sigma$. 
\end{definition}

\begin{proposition}
The space ${\cal T}_G$ is the kernel of the map ${\cal E}$.
\end{proposition}
\begin{proof}
    We note that $(x_{\delta \sigma})_e = g_e(\sigma)$ by the definitions. Then
    \begin{eqnarray*}
        {\cal E} \cdot T = 0 
        &\iff& \sum_\sigma T(\sigma) {\cal E} \cdot [\sigma] = 0 \cr
        &\iff& \sum_\sigma T(\sigma) x_{\delta \sigma} = 0 \cr
        &\iff& \sum_\sigma T(\sigma) (x_{\delta \sigma})_e = 0 \qquad \forall e \cr
        &\iff& \sum_\sigma T(\sigma) g_e(\sigma) = 0 \qquad \forall e \cr
        &\iff& (g_e, T) = 0 \qquad \forall e \cr
    \end{eqnarray*}
    The result follows from Proposition \ref{inner}.
\end{proof}
In summary, we have translated the problem of characterizing entropy sums with cancelling divergences to the problem of characterizing maps from subsets of vertices to $\mathbb{R}$ that lie in the orthogonal complement of the vectors $\{g_e\}$ kernel of ${\cal E}$. 

\section{Cancellation of codimension one divergences}

In this section, we establish two results characterizing the space ${\cal T}_G = \ker {\cal E}$ by describing two different ways to construct a basis. 

\subsection{A Fourier basis for ${\cal T}_G$}

In this section, we recall that there is a version of the Fourier transform that applies to functions on the ``Boolean cube'' (maps from subsets of $n$ items to $\mathbb{R}$) and show that an orthogonal basis for ${\cal T}_G$ can be constructed very simply using the Fourier basis elements (see, e.g., \cite{ODonnell2014}).

\begin{definition}
The {\it Fourier basis} is as set of functions in ${\cal T}_\mathbb{R}$ labeled by subsets of $V$. For a subset $\sigma$, we define 
\begin{equation}
\chi_\sigma = \sum_\tau (-1)^{x_\sigma \cdot x_\tau} [\tau] \; ,
\end{equation}
where $x_\sigma \cdot x_\tau$ is the usual dot product of the characteristic vectors.
\end{definition}

\begin{proposition}(The Fourier transform)

For any $T$, we have 
\begin{equation}
\label{Fourier}
T = \sum_\sigma \hat{T}(\sigma) \chi_\sigma
\end{equation}
where
\begin{equation}
\label{defFour}
\hat{T}(\sigma) = {1 \over 2^n} \sum_\tau T(\tau) (-1)^{x_\sigma \cdot x_\tau}\; .
\end{equation}
\end{proposition}
\begin{proof}
    Plugging in the definitions for the Fourier basis elements and the inner product, we find 
\[
(\chi_\sigma, \chi_\tau) = \sum_\alpha (-1)^{(x_\sigma + x_\tau) \cdot x_\alpha} =  2^n \delta_{\sigma, \tau}
\]
where the last equivalence follows because if any component of $x_\sigma$ and $x_\tau$ differ, then $x_\sigma + x_\tau$ is 1 for that component and the sum has a factor $\sum_{(x_\alpha)_i = \pm 1} (-1)^{(x_\alpha)_i} = 0$. If $\sigma = \tau$, all terms in the sum are 1, so we get $2^n$. Thus, the functions $\chi$ provide an orthogonal basis. We can thus expand any function $T$ as in (\ref{Fourier}). Taking the inner product with $\chi_\sigma$ on both sides gives
\[
2^n \hat{T}(\sigma) = (\chi_\sigma, T) =  \sum_\tau T(\tau) (-1)^{x_\sigma \cdot x_\tau}
\]
as desired.
\end{proof}

With this background, we can state the central result for this section:
\begin{theorem}\label{thm:Fourier}
An orthogonal basis for ${\cal T}_G = ker({\cal E})$ is provided by the Fourier basis elements $\chi_\sigma$ for all $\sigma$ not corresponding to the empty set or to an edge in $G$, together with the function
\[
\chi_0 \equiv \chi_\emptyset + \sum_{e \in E} \chi_e \; . 
\]
The orthogonal complement of ${\cal T}_G$ is spanned by
\[
\{\hat{\chi}_e \equiv {1\over 2} (\chi_e - \chi_\emptyset) | e \in E \} \; . 
\]
\end{theorem}
The proof follows almost immediately from the following:
\begin{lemma}\label{lem:Eaction}
    For any $T \in {\cal T}_\mathbb{R}$ and $e \in E$, we have that
\[
({\cal E} T)_e =  2^{n-1}(\hat{T}(\emptyset) - \hat{T}(e)) = (\hat{\chi}_e, T) 
\]
\end{lemma}
\begin{proof}
For any subset $\sigma$ of vertices and for any edge $e_{ij} \in E$ between vertices $i$ and $j$, we have
\[
(\delta \sigma)_{e_{ij}} = {1 \over 2} \left( 1- (-1)^{(x_{\sigma})_i + (x_\sigma)_j} \right) \; .
\]
This follows immediately from the definition of the cut map. From this, for any function $T \in {\cal T}_\mathbb{R}$, we have 
\[
({\cal E} T)_{e_{ij}} = \sum_{\sigma} T(\sigma) ({\cal E} [\sigma])_{e_{ij}} = \sum_{\sigma} T(\sigma) (\delta \sigma)_{e_{ij}} = {1 \over 2} \sum_{\sigma} T(\sigma) \left( 1- (-1)^{(x_{\sigma})_i + (x_\sigma)_j} \right)
\]
Recalling the result (\ref{defFour}) for the coefficients $\hat{T}$ of $T$ in the Fourier basis, we recognize the first term as $2^{n-1} \hat{T}(\emptyset)$ and the second as $-2^{n-1} \hat{T}(e_{ij})$ so the result follows.
\end{proof}
Using this lemma, we now prove Theorem \ref{thm:Fourier}.
\begin{proof} (Theorem \ref{thm:Fourier})
From Lemma \ref{lem:Eaction}, we see that ${\cal E} T$ vanishes if and only if for every edge $e \in E$, the Fourier coefficient associated with $\chi_e$ is the same as the Fourier coefficient associated with $\chi_\emptyset$. In the space ${\cal T}_{E \emptyset}$ spanned by $\chi_\emptyset$ and $\{\chi_e | e \in E\}$, only the vector $\chi_0$ defined in the statement of the theorem satisfies this condition. The orthogonal complement of ${\cal T}_{E \emptyset}$ is spanned by $\chi_\sigma$ for which $\sigma$ is neither the empty set nor an edge in $G$. These all  satisfy our condition, are orthogonal to one another, and are orthogonal to $\chi_0$, so these $\chi$ together with $\chi_0$ span ${\cal T}_G$. The orthogonal complement  $({\cal T}_G)^\perp$ is the subspace of ${\cal T}_{E \emptyset}$ orthogonal to $\chi_0$. We can check that the functions $\hat{\chi}_e = (\chi_e - \chi_\emptyset)/2$ for edges $e$ are each orthogonal to $\chi_0$; the matrix of their inner products has $2^{n-1}$ for each diagonal element and $2^{n-2}$ for all off diagonal elements. This is non-singular, so $\{\hat{\chi}_e\}$ span $({\cal T}_G)^\perp$.
\end{proof}

\begin{example}
As an example, consider the graph with three vertices {1,2,3} connected as $1-2-3$. Then the basis elements and corresponding combination of entropies are 
\begin{equation}
\begin{array}{ll}
    \chi_1 & \qquad S_\emptyset - S_1 + S_2 + S_3 - S_{12} - S_{13} + S_{23} - S_{123}\cr
    \chi_2 & \qquad S_\emptyset + S_1 - S_2 + S_3 - S_{12} + S_{13} - S_{23} - S_{123} \cr
    \chi_3 & \qquad S_\emptyset + S_1 + S_2 - S_3 + S_{12} - S_{13} - S_{23} - S_{123}\cr
    \chi_{13} & \qquad S_\emptyset - S_1 + S_2 - S_3 - S_{12} + S_{13} - S_{23} + S_{123}\cr
    \chi_{123} & \qquad S_\emptyset - S_1 - S_2 - S_3 + S_{12} + S_{13} + S_{23} - S_{123}\cr
   \chi_{\emptyset} + \chi_{12} + \chi_{23}& \qquad 3 S_\emptyset + S_1 - S_2 + S_3 + S_{12} - S_{13} + S_{23}+ 3 S_{123}\cr
\end{array}
\end{equation}
Each of these has a nonzero coefficient for every subset of regions. More generally, each of the entropy sums for the case with $n$ regions will have $2^n$ nonvanishing terms, so the individual basis elements are somewhat unwieldy and the corresponding entropy sums do not have a simple information theoretic interpretation. In the next section, we will see that it is always possible to form a basis using only sums of the four types described in the introduction. For any number of regions, these have either 1,2,3, or 7 nonvanishing terms and have standard information theory interpretations.
\end{example}

\subsection{``Information'' bases for ${\cal T}_G$.}

\begin{figure}
    \centering
\includegraphics[width=0.95\textwidth]{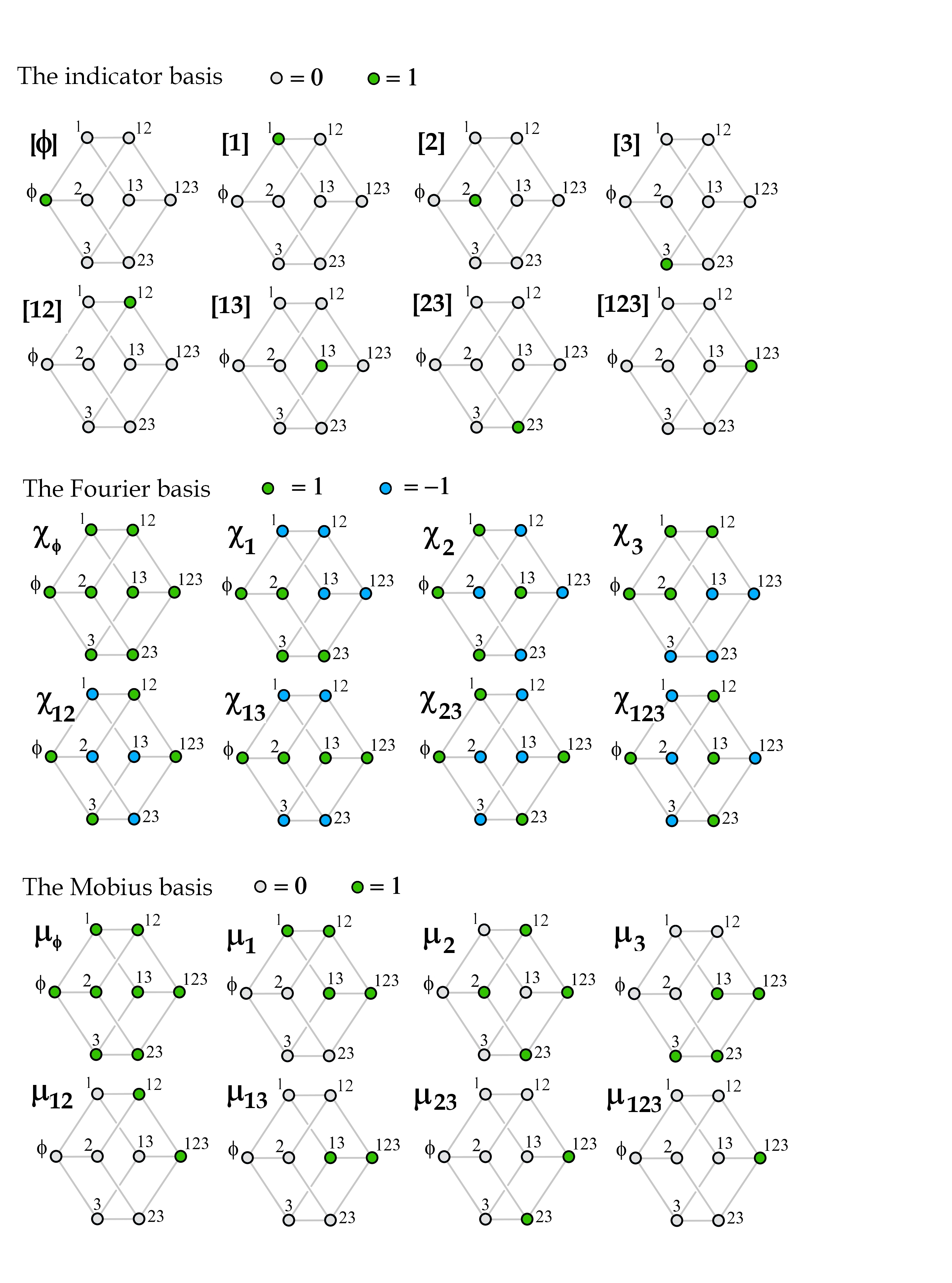}
    \caption{Indicator, Fourier, and Möbius bases for functions on subsets of a three-element set $\{1,2,3\}$ (the ``Boolean cube'').}
    \label{fig:Boolean1}
\end{figure} 

We now move on to the statement and proof of a theorem that provides an alternative spanning set of functions for ${\cal T}_G$ whose corresponding entropy sums have a simpler information-theoretic interpretation. We begin by defining four classes of functions
\begin{definition}\label{def:T14}
\begin{enumerate}
    \item Let $T_0$ be the function that assigns 1 to the empty set and 0 to all non-empty subsets.
    \item Let $T_\sigma = [\sigma] - [\sigma^c]$ for any subset $\sigma$.
    \item Let $T_{\sigma,\tau} = [\sigma] + [\tau] - [\sigma \cup \tau]$ for any disjoint $\sigma$ and $\tau$ whose induced subgraphs are not connected by any edge in $G$. 
    \item Let $T_{\sigma,\tau,\upsilon} = [\sigma] + [\tau]  + [\upsilon] - [\sigma \cup \tau ] - [\sigma\cup \upsilon] -[ \tau \cup \upsilon] + [\sigma \cup \tau \cup \upsilon]$ for any disjoint $\sigma$, $\tau$, and $\upsilon$. 
\end{enumerate}
\end{definition}
Our main result for this section is that these four types of functions span all of ${\cal T}_G$:
\begin{theorem}
\label{thm:Möbius}
    Define ${\cal T}_G^*$ to be span of the four types of functions from Definition \ref{def:T14}. Then ${\cal T}_G^* = {\cal T}_G$. 
\end{theorem}
The theorem follows from a pair of lemmas:
\begin{lemma}
\label{lem:inkernel}
    The subspace ${\cal T}_G^*$ is in $ker({\cal E})$. 
\end{lemma}
\begin{lemma}
\label{complemma}
The orthogonal complement of ${\cal T}_G^*$ is in the orthogonal complement of $ker({\cal E})$.
\end{lemma}
We begin with the proof of Lemma \ref{lem:inkernel}
\begin{proof} (Lemma \ref{lem:inkernel})
We show that functions of each type are in $ker({\cal E})$:
\vspace{-0.3cm}
    \begin{enumerate}
    \item There are no edges in the cut of the empty set, so ${\cal E} T_0 = 0$.
    \item The function $T_\sigma$ is in $\ker {\cal E}$ since $e$ is in the cut for $\sigma$ $e$ if and only if $e$ is in the cut for $\sigma^c$.
    \item For subsets $\sigma$ and $\tau$ not connected by an edge, an edge will be in the cut of $\sigma \cup \tau$ if and only if it is in the cut of $\sigma$ or $\tau$. Further, no edge can be in the cut of both $\sigma$ and $\tau$, so ${\cal E} T_{\sigma, \tau} = 0$.
    \item  To show that $T_{\sigma, \tau, \upsilon}$ is in $\ker {\cal E}$, consider any edge in the cut of exactly one the three elementary subsets, say $\sigma$. Then it will also be in the cuts for $\sigma \cup \tau$, $\sigma \cup \upsilon$, and $\sigma \cup \tau \cup \upsilon$ but none of the other cuts for the subsets appearing in $T_{\sigma, \tau, \upsilon}$. Thus $({\cal E} T_{\sigma, \tau, \upsilon})_e$ will vanish. If $e$ is in the cuts for two of the elementary subsets, say $\sigma$ and $\tau$, it will also be in the cuts for $\sigma \cup \upsilon$ and $\tau \cup \upsilon$ and none of the others, so again $({\cal E} T_{\sigma, \tau, \upsilon})_e$  vanishes. Finally, if $e$ is not in any of the cuts for any of the elementary subsets, it won't be in any of the other cuts either, so again $({\cal E} T_{\sigma, \tau, \upsilon})_e$ vanishes.
\end{enumerate}
\end{proof}
In order to prove Lemma \ref{complemma}, it will be convenient to define a third basis for ${\cal T}_\mathbb{R}$.
\begin{definition}
The {\it Möbius basis} is a set of functions in ${\cal T}_\mathbb{R}$ labeled by subsets of $V$ (see \cite{Rota1964,Stanley2011}). For a subset $\sigma$, we define 
\[
\mu_\sigma = \sum_{\tau \supseteq \sigma} [\tau] 
\]
\end{definition}
\begin{proposition}(The Möbius transform)
For any $T$, we have 
\begin{equation}
\label{Möbius}
T = \sum_\sigma \tilde{T}(\sigma) \mu_\sigma
\end{equation}
where
\begin{equation}
\label{defMob}
\tilde{T}(\sigma) = \sum_{\tau \subseteq \sigma}(-1)^{|\sigma| - |\tau|} T(\tau) \; .
\end{equation}
\end{proposition}
\begin{proof}
    The result follows from the identity
    \[
    \sum_{\sigma | \tau \subseteq \sigma \subseteq \upsilon \ } (-1)^{|\sigma| - |\tau|}  = \delta_{\tau, \upsilon} \; .
    \]
    This is immediate for $\tau = \upsilon$. If $\tau$ is a proper subset such that $\upsilon$ has $k$ additional elements, we have $\binom{k}{l}$ elements $\sigma$ satisfying $\tau \subseteq \sigma \subseteq \upsilon$ with $l$ more elements that $\tau$, so 
    \[
    \sum_{\sigma | \tau \subseteq \sigma \subseteq \upsilon \ } (-1)^{|\sigma| - |\tau|} = \sum_{l=0}^k (-1)^l \binom{k}{l} = (1 - 1)^k = 0.
    \]
\end{proof}
To prove Lemma \ref{complemma}, we will show that orthogonality to the various types of functions in ${\cal T}_G^*$ imposes simple conditions on the coefficients of a function $T$ in the Möbius basis, ultimately leaving us with a simple expression that is a linear combination of the functions $\hat{\chi}_e$ that span $(\ker {\cal E})^\perp$.
\begin{proof} (Lemma \ref{complemma})
Suppose that $T \in ({\cal T}_G^*)^\perp$. Then for $S \in {\cal T}_G^*$, we have 
    \begin{eqnarray*}
    0 &=& (S,T) = \sum_\chi S(\chi) T(\chi) \cr
    &=& \sum_\chi S (\chi) \sum_\eta \tilde{T}(\eta) \mu_\eta(\chi) = \sum_\chi S (\chi) \sum_{\eta \subseteq \chi} \tilde{T}(\eta) \cr
     &=&  \sum_\eta \tilde{T}(\eta) \sum_{\chi \supseteq \eta} S (\chi) \cr
    \end{eqnarray*}
where in the second line, we have expanded $T$ in terms of the Möbius basis. 

We now examine the consequences of this for various specific types of functions in ${\cal T}_G^*$.

1) Taking $S = T_0$, the only nonzero term in the second sum above is with $\chi = \emptyset$, so we have:
\[
\tilde{T}(\emptyset) = 0 \; .
\]

2) Taking $S = T_{\sigma, \tau ,\upsilon}$, we get 
\[
0 =  \sum_\eta \tilde{T}(\eta) ({\bf 1}_{\eta \subseteq  \sigma}+ {\bf 1}_{\eta \subseteq  \tau} + {\bf 1}_{\eta \subseteq  \upsilon} - {\bf 1}_{\eta \subseteq  \sigma \cup \tau } -{\bf 1}_{\eta \subseteq \tau \cup \upsilon} - {\bf 1}_{\eta \subseteq  \sigma \cup \upsilon} +
{\bf 1}_{\eta \subseteq  \sigma \cup \tau \cup \upsilon} )
\]
Here, the ${\bf 1}$s are 1 if the inclusion holds and zero otherwise. The sum of ${\bf 1}$s 1 if $\eta = \emptyset$ or $\eta$ is a subset of $\sigma \cup \tau \cup \upsilon$ that has overlap with all three sets, and 0 otherwise. Thus, using $\tilde{T}(\emptyset) = 0$, we have
\begin{equation}
\label{zerosum}
0 =  \sum_{\alpha \subset^* \sigma} \sum_{\beta \subset^* \tau} \sum_{\gamma \subset^* \upsilon} \tilde{T}(\alpha \cup \beta \cup \gamma)
\end{equation}
where the *s indicate that the subset must be non-empty.
Now, for any three distinct vertices $\{i,j,k\}$, taking $\sigma=\{i\}$, $\tau = \{j\}$, and $\upsilon = \{k\}$ gives 
\[
\tilde{T}(\{i,j,k\}) = 0 \; .
\]
Supposing that we have shown all $\tilde{T}$s vanish for subsets from 3 element up to $m-1$ elements, consider an $m$-element subset $\{k_1, \dots, k_m\}$. Taking $\sigma=\{k_1\}$, $\tau = \{k_2\}$, and $\upsilon = \{k_3, \dots, k_m\}$, we get from (\ref{zerosum}) that 
\[
0 = \tilde{T}(\{k_1, \dots, k_m\}) 
\]
since all other terms on the right hand side have already been shown to vanish. 

We have now shown that all Möbius coefficients corresponding to subsets with size 0 or size $n \ge 3$ vanish, so we can write
\begin{equation}
\label{shortsum}
T = \sum_i \tilde{T}(\{i\}) \mu_{\{i\}} +\sum_{i < j} \tilde{T}(\{i,j\}) \mu_{\{i,j\}}
\end{equation}

3) Taking $S = T_{\{i\},\{j\}}$ where $ij$ is not an edge gives 
\[
0 = -\tilde{T}(\{i,j\}) 
\]
since the terms with $\tilde{T}(\{i\})$ and $\tilde{T}(\{j\})$ each have two contributions that add to zero. Thus, the sum over two-element sets in (\ref{shortsum}) is restricted to edges.

4) Finally, consider $S = T_{\sigma} = [\sigma] - [\sigma^c]$. Orthogonality with $S$ gives
\[
0 = T(\sigma) - T(\sigma^c) \; .
\]
To see the implications for the Möbius coefficients in (\ref{shortsum}) it will be useful to note that
\begin{eqnarray*}
\mu_{i}(\sigma) &=& (x_\sigma)_i \cr
\mu_{i}(\sigma^c) &=&  1-(x_\sigma)_i \cr
\mu_{ij}(\sigma) &=&  (x_\sigma)_i(x_\sigma)_j \cr
\mu_{ij}(\sigma^c) &=&  (1-(x_\sigma)_i)(1-(x_\sigma)_j)
\end{eqnarray*}
Rewriting the constraint $T(\sigma) = T(\sigma^c)$ by making use of these expressions in (\ref{shortsum}) we find
\[
\sum_i \tilde{T}(\{i\}) (x_\sigma)_i +\sum_{\{i,j\} \in E} \tilde{T}(\{i,j\}) (x_\sigma)_i(x_\sigma)_j = \sum_i \tilde{T}(\{i\}) ( 1-(x_\sigma)_i ) +\sum_{\{i,j\} \in E} \tilde{T}(\{i,j\})  (1-(x_\sigma)_i)(1-(x_\sigma)_j)
\]
The terms quadratic in $x$ already match between the two sides. In order that the terms linear in $x$ match, we must have for every vertex $i$, 
\[
2\tilde{T}(\{i\}) + \sum_{j |\{i,j\} \in E} \tilde{T}(\{i,j\}) = 0
\]
This also ensures (by summing over $i$) that the $x$-independent terms match. With the new constraint, we have that $T$ can be written as
\begin{equation}
\label{Tfin}
T(\sigma) = \sum_{\{i,j\} \in E} \tilde{T}(\{i,j\}) ((x_\sigma)_i(x_\sigma)_j - (x_\sigma)_i/2 - (x_\sigma)_j/2)
\end{equation}
In Theorem \ref{thm:Fourier}, we showed that the orthogonal complement of ${\cal T}_G$ is spanned by functions $\hat{\chi}_e = (\chi_e - \chi_\emptyset)/2$ for $e \in E$. If $e = \{i,j\}$, we note that
\[
\hat{\chi}_e(\sigma) = (\chi_e(\sigma) - \chi_\emptyset(\sigma))/2 = {1\over 2} ((-1)^{(x_\sigma)_i + (x_\sigma)_j} - 1) = 2 [(x_\sigma)_i (x_\sigma)_j - (x_\sigma)_i/2 - (x_\sigma)_j/2)]  \; ,
\]
recalling that the $x$s only take values $0,1$. Thus, the result (\ref{Tfin}) is precisely the statement that $T$ is in the span of $\hat{\chi}_e$, or $T \in {\cal T}_G^\perp$. This completes the proof that $({\cal T}_G^*)^\perp \subset (\ker {\cal E})^\perp$.
\end{proof}
The details of the proof here suggest one way to choose a particular basis of ${\cal T}_H$ from among the functions in Definition \ref{def:indicator}:
\begin{proposition}
    Choose an ordering of vertices and let $\sigma_k$ denote the set containing the $k$th smallest element of a subset $\sigma$ according to the ordering. Then the following functions provide a basis for ${\cal T}_G$: 1) $T_0$  2) $T_{\{i\}}$ for all vertices $i$ 3) $T_{\{i,j\}}$ for non-edges $\{i,j\}$, and 4) $T_{\sigma_1,\sigma_2, \sigma \backslash (\sigma_1 \cup \sigma_2)}$ for each $\sigma$ with $|\sigma| \ge 3$.
\end{proposition}
\begin{proof}
    Each of the functions here is in ${\cal T}_G$ and their number is $1 + n + (\binom{n}{2} - |E|) + (2^n - \binom{n}{2} - n - 1) = 2^n - |E| = \dim{\cal T}_G$. To show they form basis, we need to show that they are independent. If they were not, the orthogonal complement of their span would be larger than the orthogonal complement of ${\cal T}_G$. But the preceding proof shows that the orthogonal complement of the span of these functions is in the orthogonal complement of ${\cal T}_G$, since in that proof, we only needed to use orthogonality to functions of the type described here. In particular, for step 4, orthogonality with $T_{\{i\}}$ for the $n$ vertices $i$ is enough to give the $n$ equations associated with vanishing of linear terms in $x$. 
\end{proof}

\section{Cancelling divergences from multiple intersections}

So far, we have been concerned about cancellation of divergences related to the pairwise intersections of regions. Our considerations apply equally well to the leading area-law divergences and to any subleading divergences or divergences associated with non-smooth features of these boundary regions, for example corners in a one dimensional boundary between a pair of regions in two dimensions. To see this, we can use the fact that the divergences appearing in entropies are independent of the state (they arise from the UV degrees of freedom in the quantum field theory whose entanglement structure is the same in different finite-energy states). For a pure state of the whole system, the entropy of a region $A$ equals the entropy of the complementary region $A^c$, so in particular, all divergences associated with the boundary of $A$ must match the divergences associated with the boundary of $A^c$. Thus, our condition that the sum of coefficients for all regions having a particular boundary component (regardless of which side of the boundary the region occupies) should ensure that all types of divergences associated with this boundary component will cancel. 

Additional divergences may arise from the intersection of $m>2$ regions at some higher codimension feature $J$ \cite{CasiniHuerta2007_Universal2p1,CasiniHuerta2009_Review,BuenoMyersWitczak2015_PRL,BuenoMyers2015_JHEP,BuenoMyersWitczak2015_Twist,KallinEtAl2014_JStatMech,StoudenmireEtAl2014_PRB,HaywardSierensEtAl2017_TrihedralPRB,BednikEtAl2019_TrihedralPRB,SeminaraSistiTonni2017_BCFT,Berthiere2019_BoundaryCornerPRB,Nishioka2018_RMP}. Let $I_m$ be the subset of vertices corresponding to the regions $A_i$ whose boundaries intersect at $J$. For any non-empty proper subset $\tau \subset I_m$, the region $A_\tau$ has a potential divergence associated with the ``corner'' of $A_\tau$ that touches $J$.\footnote{It is possible that no such divergence appears: for example, with three two-dimensional regions intersecting at a point, one of the three angles might be $\pi$ in which case there is no corner divergence for the corresponding region. It may also be that the corner divergence is absent if the corner is sufficiently ``cuspy'' \cite{BuenoCasiniWitczak2019_SingularJHEP}.} In an entropy sum, this divergence (if it exists) will be cancelled provided that the sum of coefficients for sets $\sigma$ containing $\tau$ but not $I_m \backslash \tau$ (the complement of $\tau$ in $I_m$) plus the sum of coefficients for sets $\sigma$ containing $I_m \backslash \tau$ but not $\tau$ vanishes. Here, we are using the fact that the divergence associated with a corner will be the same as the divergence associated with its complementary corner. As above, this follows since divergences are state independent, and the entropy of a subsystem is the same as the entropy of its complement in a pure state. 

We can generalize our earlier analysis in order to analyze these additional conditions associated with multiple intersections. In order to capture the information about all intersections, we can generalize from a graph to a {\it hypergraph} $H$, which is again a set $H_1 \equiv V$ of $n$ vertices associated with the $n$ regions, but now with sets $H_2 \equiv E, H_3, H_4 \dots$, of pairs, triples, quadruples, etc... of vertices corresponding to the intersections of 3,4, etc... regions. We define ${\cal I}_H = \cup_{m \ge 2} H_m$ and refer to this as the set of all {\it intersections}.

Starting from the space ${\cal T}_\mathbb{R}$ of functions on subsets of the vertices, we can now define a subspace ${\cal T}_H$ of functions that are safe according to the conditions above:
\begin{definition}
    Let $H$ be a hypergraph and ${\cal T}_{\mathbb{R}}$ be the space of real functions on subsets of its vertices. We define the subspace ${\cal T}_H$ of {\it safe} functions by the condition that for each intersection $I \in {\cal I}_H$ and each bipartition $\{a,b\}$ of $I$, 
\begin{equation}
\label{gencond}
\sum_{\sigma | a \in \sigma, b \notin \sigma} T(\sigma) + \sum_{\sigma | a \notin \sigma, b \in \sigma} T(\sigma) = 0 \; .
\end{equation}
\end{definition}
Equivalently, defining functions
\[
g_{a,b}(\sigma) = \left\{ \begin{array}{ll} 1 & (a \cup b) \cap \sigma \in \{a,b\} \cr
0 & \qquad {\rm otherwise} \end{array}  \right.
\]
for nonintersecting sets $a,b$, ${\cal T}_H$ is the subspace defined by $(g_{a,b}, T) = 0$ for all bipartitions $\{a,b\}$ of all intersections $I$. 

\subsection{A Fourier basis for ${\cal T}_H$}

We will now see that, as for the edge constraints, the constraints related to general intersections may be expressed very simply in the Fourier basis:

\begin{lemma}
\label{FourLemma}
For a general hypergraph $H = \{H_m\}$, a function $T \in {\cal T}_\mathbb{R}$ is in ${\cal T}_H$ if and only if for any intersection $I \in {\cal I}_H$ and any even order subset $\tau \subset I$, the Fourier coefficient $\hat{T}(\tau)$ is equal to $\hat{T}(\emptyset)$. 
\end{lemma}
This immediately leads (following the same logic as in the proof of Theorem \ref{thm:Fourier}) to
\begin{theorem}\label{thm:Fourier2}
For a hypergraph $H$, define $I_E$ to be the collection of even order subsets of vertices (including the empty set) that are subsets of some intersection $I \in {\cal I}_H$. Then an orthogonal basis for ${\cal T}_H$ is provided by the Fourier basis elements $\chi_\sigma$ for $\sigma \notin I_E$ together with the function
\[
\chi_0 \equiv \sum_{\sigma \in I_E} \chi_\sigma \; . 
\]
The dimension of ${\cal T}_H$ is $2^n - |I_E| + 1$.
\end{theorem}

\begin{proof} (Lemma \ref{FourLemma})
Consider an intersection $I$ and some bipartition $\{a,b\}$ of $I$.
Expressed in terms of the Fourier basis, the condition (\ref{gencond}) for the pair $\{a,b\}$ is
\begin{equation}
\label{multicond}
\sum_\sigma \hat{T}(\sigma) (g_{a,b}, \chi_\sigma) = 0 \; ,
\end{equation}
where
\[
(g_{a,b}, \chi_\sigma) = \sum_\tau g_{a,b}(\tau)  (-1)^{x_\sigma \cdot x_\tau} \; .
\]
Now, for any vertex $i$ not in $I$ and any set $\tau$ for which $g_{a,b} \ne 0$, there is a distinct complementary set $\tau'$ with $g_{a,b} \ne 0$ obtained by removing / adding the element $i$ if it is present / not present. If $i$ is in $\sigma$, then the terms in the sum above corresponding to complementary sets $\tau$ and $\tau'$ will cancel each other, so the sum vanishes. Thus, inner prduct above is nonzero only for $\sigma \subset I$. 

The terms in the sum for which $g_{a,b}(\tau)$ is nonvanishing are of the form $\tau = a \cup \varphi$ or $\tau = b \cup \varphi$ for $\varphi \cap I = \emptyset$. So we have, for $\sigma \subset I$,
\[
(g_{a,b}, \chi_\sigma) = \sum_\varphi (-1)^{x_\sigma \cdot x_{a \cup \varphi}} + (-1)^{x_\sigma \cdot x_{b \cup \varphi}} .
\]
Since $\sigma \cap \varphi = \emptyset$, we have $x_\sigma \cdot x_{a \cup \varphi} = x_\sigma \cdot x_{a}$, and $x_\sigma \cdot x_{b \cup \varphi} = x_\sigma \cdot x_{b}$. Thus, for any $\sigma$,
\[
(g_{a,b}, \chi_\sigma) =  \left\{ \begin{array}{ll} 2^{n-|I|} ((-1)^{x_\sigma \cdot x_a} + (-1)^{x_\sigma \cdot x_b}) &  \qquad \sigma \subset I \cr
0 & \qquad {\rm otherwise} \end{array}  \right.
\]
It follows that the condition (\ref{multicond}) involves only Fourier coefficients $\hat{T}_\sigma$ for sets $\sigma \subset I$, reducing to 
\begin{equation}
\label{multicond2}
\sum_{\sigma \subset I} \hat{T}(\sigma) ((-1)^{x_\sigma \cdot x_a} + (-1)^{x_\sigma \cdot x_b}) = 0 \; .
\end{equation}
The expression in brackets is nonzero if and only if $a \cap \sigma$ and $b \cap \sigma$ have the same parity, which is equivalent to $|\sigma| = |(a \cap \sigma) \cup (b \cap \sigma)|$ being even. We can now further reduce the condition to
\[
\sum_{\sigma \subset I, |\sigma| \; even}  \hat{T}(\sigma) (-1)^{x_\sigma \cdot x_a} = 0
\]
Since we have one equation for each bipartition of $I$, this is a set of $2^{|\sigma|-1}-1$ equations for $2^{|\sigma|-1}$ unknowns. We will show that these equations are all independent, leaving a one-dimensional space of solutions in which all $\hat{T}$s appearing here are equal. 

Choose some arbitrary element $1 \in I$ and let $\tilde{I} = I \backslash 1$. Then bipartitions $\{a,b\}$ are in one-to-one correspondence with non-empty subsets $c \in \tilde{I}$ via $a= c, b = I \backslash c$. Define a bijection $e$ from subsets of $\tilde{I}$ to even order subsets of $I$ by 
\[
e(\sigma) = \left\{ \begin{array}{ll} \sigma & \qquad |\sigma| \; even \cr
{1} \cup \sigma & \qquad |\sigma| \; odd \cr
\end{array}
\right.
\]
Then our condition is equivalent to the statement that for every $c \in \tilde{I}$
\begin{equation}
\label{reduced}
\sum_{\sigma \subset \tilde{I}} \hat{T}(e(\sigma)) (-1)^{y_\sigma \cdot y_c} = 0
\end{equation}
where $y_\sigma \in \mathbb{R}^{|I| - 1}$ is the characteristic function for a subset $\sigma \in \tilde{I}$. Defining $\{\tilde{\chi}_\sigma\}$ to be the Fourier basis for functions on subsets of $\tilde{I}$ and defining a function $C$ on subsets of $\tilde{I}$ by
\[
C = \sum_\sigma \hat{T}(e(\sigma)) \tilde{\chi}_\sigma \; ,
\]
the condition (\ref{reduced}) is simply that $C(c) = 0$ for all non-empty $c \subset \tilde{I}$. This is equivalent to the equality of all the Fourier coefficients $\hat{T}(e(\sigma))$, which are the Fourier coefficients associated with the even subsets of $I$. 
\end{proof}
\begin{remark}
    For the case where $H$ is an ordinary graph, theorem $(\ref{thm:Fourier2})$ reduces to (\ref{thm:Fourier}): in this case intersections are edges, and the set $I_E$ is the empty set together with the set of edges.
\end{remark}

\subsection{Divergence-free mutual information and tripartite information}

In this section, we consider the information basis functions defined in Definition \ref{def:T14} and investigate which of these is safe from divergences associated with higher codimension intersections.

To state the main result, it will be convenient to define an appropriate generalization of the notion of subsets being connected by an edge. 
\begin{definition}
Given a set $\alpha \in I_E$ (i.e. any even order subset of some intersection $I \in {\cal I}_H$) we say that a collection of disjoint subsets $\{\sigma_i\}$ of vertices is $\alpha$-odd if the intersection $\sigma_i \cap \alpha$ has odd order for all $i$. We say the collection is $I_E$-adjacent if it is $\alpha$-odd for some $\alpha \in I_E$.
\end{definition}
The following gives an equivalent criterion for $I_E$-adjacency:
\begin{proposition}
\label{prop:adjacent}
    A collection $\{\sigma_i\}$ of $m$ mutually disjoint subsets is $I_E$-adjacent if and only if one of the following holds:
    
    i) $m$ is even and there exists $I \in {\cal I}_H$ intersecting each of $\{\sigma_i\}$. 
    
    ii) $m$ is even and there exists $I \in {\cal I}_H$ intersecting each of $\{\sigma_i\}$ and another region $\tau$ disjoint from each $\sigma_i$. 
\end{proposition}
\begin{proof}
    If $\{\sigma_i\}$ is $I_E$-adjacent, there is $\alpha \in I_E$ such that $|\alpha \cap \sigma_i|$ is odd for each $i$. It follows that $\alpha \cap \sigma_i \ne \emptyset$ for each $i$. For odd $m$, the union of all $\alpha \cap \sigma_i$ has odd order. Since $|\alpha|$ is even, there must be an element of $\alpha$ not in any $\sigma_i$. Call the set with this element $\tau$.
    The set $\alpha$ is part of some intersection $I \in {\cal I}_H$ by the definition of $I_E$. This $I$ intersects all $\sigma_i$ in the even case and all $\sigma_i$ plus $\tau$ in the odd case. Conversely, if such an $I$ exists, form $\alpha$ by choosing one element of $\sigma_i \cap I$ for each $i$ and additionally an element of $\tau \cap I$ in the odd case. The collection $\{\sigma_i \}$ is $\alpha$-odd in both cases, so also $I_E$-adjacent
\end{proof}
We can now define a subset of the functions in Definition \ref{def:T14} that we will show is the subset belonging to ${\cal T}_H$.
\begin{definition}\label{def:H4}
For any hypergraph $H$ with even intersection subsets $I_E$, define the following functions on subsets of the vertices:
\begin{enumerate}
    \item Let $T_0$ be the function that assigns 1 to the empty set and 0 to all other subsets.
    \item Let $T_\sigma = [\sigma] - [\sigma^c]$ for any subset $\sigma$.
    \item Let $T_{\sigma,\tau} = [\sigma] + [\tau] - [\sigma \cup \tau]$ for any non-$I_E$-adjacent disjoint $\sigma, \tau$. 
    \item Let $T_{\sigma,\tau,\upsilon} = [\sigma] + [\tau]  + [\upsilon] - [\sigma \cup \tau ] - [\sigma\cup \upsilon] -[ \tau \cup \upsilon] + [\sigma \cup \tau \cup \upsilon]$ for any pairwise disjoint, non-$I_E$-adjacent $\sigma$, $\tau$, and $\upsilon$. 
\end{enumerate}
\end{definition}
With these definitions, we have:
\begin{proposition}
\label{thm:Möbius2}
    The functions in Definition \ref{def:H4} are precisely the functions in Definition \ref{def:T14} belonging to ${\cal T}_H$.
\end{proposition}
\begin{proof}
    For an subset $\alpha \in I_E$, it will be convenient to define the function $\Delta_\alpha : {\cal T}_\mathbb{R} \to \mathbb{R}$ by
    \[
    \Delta_\alpha(T) = \hat{T}(\alpha) - \hat{T}(\emptyset) \; .
    \]
    If $G$ is the ordinary graph $\{V = H_1, E = H_2\}$, the space ${\cal T}_H$ is the subspace of ${\cal T}_G$ satisfying the additional condition that $\Delta_\alpha(T) = 0$ for any $\alpha \in I_E \backslash E$. For the four types of functions in Definition \ref{def:H4}, not yet imposing new the non-adjacency restrictions, a short calculation using (\ref{defFour}) shows that for any $\alpha$
    \begin{eqnarray*}
       \Delta_\alpha(T_0) &=& 0 \cr
       \Delta_\alpha(T_\sigma) &=& 0 \cr
       \Delta_\alpha(T_{\sigma,\tau}) &=& -{1 \over 2^n}(1 - (-1)^{x_\alpha \cdot x_\sigma})(1 - (-1)^{x_\alpha \cdot x_\tau}) \cr
       \Delta_\alpha(T_{\sigma,\tau,\upsilon}) &=& -{1 \over 2^n}(1 - (-1)^{x_\alpha \cdot x_\sigma})(1 - (-1)^{x_\alpha \cdot x_\tau})(1 - (-1)^{x_\alpha \cdot x_\upsilon})  
    \end{eqnarray*}    
    For functions $T_0$ and $T_\sigma$, $\Delta$ always vanishes, so all of these are in ${\cal T}_H$. The third expression vanishes if and only if one of $x_\sigma \cdot x_\alpha$ and $x_\tau \cdot x_\alpha$ is even. This will be true for every $\alpha \in I_E$ if and only if $\sigma$ and $\tau$ are non-$I_E$-adjacent. Similarly, given disjoint $\sigma$, $\tau$, and $\upsilon$, the fourth expression vanishes for every $\alpha \in I_E$ if and only if one of $x_\sigma \cdot x_\alpha$, $x_\tau \cdot x_\alpha$, and $x_\upsilon \cdot x_\alpha$ is even, so these three sets are non-$I_E$-adjacent. Thus, the functions in Definition \ref{def:H4} are exactly the subset of functions in Definition \ref{def:T14} satisfying the additional constraint $\Delta_\alpha(T) = 0$. 
\end{proof}

\begin{remark} (Implications for mutual information)
    For the mutual information $S(A_\sigma, A_\tau)$, the previous condition that $\sigma$ and $\tau$ are not connected by an edge translates to the requirement that $A_\sigma$ and $A_\tau$ do not share a codimension-one boundary. Proposition \ref{prop:adjacent} generalizes this to require that no subset $\{ij\}$ of any intersection has size 1 overlap with both $\sigma$ and $\tau$. This translates to the condition that the boundaries of $A_\sigma$ and $A_\tau$ cannot have any intersection, even of lower codimension (see however Remark \ref{caveat}). 
\end{remark}
\begin{remark} (Implications for tripartite information)
    For the tripartite information, the participating regions were previously unconstrained, but our new condition translates to the statement that  $A_\sigma$, $A_\tau$, and $A_\upsilon$ must not meet at any intersection that also involves a fourth region.
\end{remark}
Figure \ref{fig:Examples} shows examples of entropy combinations for which the divergences associated with codimension one boundaries cancel but divergences associated with higher-codimension intersections do not.
\begin{remark}
\label{caveat}
    The space of safe functions ${\cal T}_H$ corresponds to entropy sums that we expect to be finite for any collection of regions with the connectivity properties specified by $H$. However, there may be special cases where a potentially divergent entropy combination corresponding to a function not in ${\cal T}_H$ is actually finite. An example shown in Figure \ref{fig:Examples} (middle). Here, two regions intersect, but the intersection angle is zero; \cite{BuenoCasiniWitczak2019_SingularJHEP} showed that in some situations like this, the divergence can be absent. 
\end{remark}

\begin{figure}
    \centering
\includegraphics[width=0.9\textwidth]{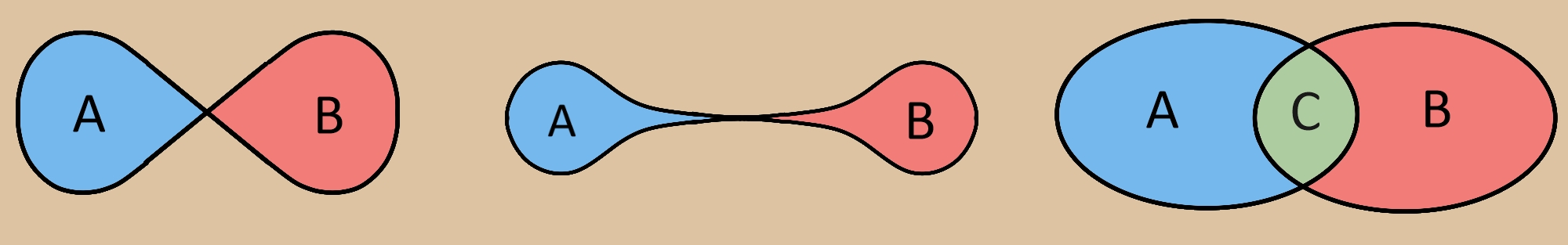}
    \caption{Left: mutual information $I(A,B)$ has cancelling divergences associated with codimension one boundaries, but uncancelled corner divergences. Middle: for corners that are sufficiently cuspy, these additional divergences may be absent. Right: the tripartite information $I(A,B,C)$ has cancelling divergences associated with codimension one boundaries, but uncancelled corner divergences.}
    \label{fig:Examples}
\end{figure} 
It is natural to ask whether there is a result analogous to Theorem \ref{thm:Möbius2}, that is, whether the functions in \ref{def:H4} span all of ${\cal T}_H$. It turns out that this is not true in general.
\begin{proposition}
    Let ${\cal T}_H^*$ be the span of the functions in Definition \ref{def:H4}. For certain cases, ${\cal T}_H^*$ is strictly smaller than ${\cal T}_H$.
\end{proposition}
\begin{proof}
We construct a specific example. Consider the hypergraph with $n$ vertices and all possible intersections of orders 2 and 4, with no other intersections. Then all disjoint pairs $\sigma,\tau$ will be $I_E$ adjacent, since taking an arbitrary element of each gives a set $\alpha \in I_E$. So no $T_{\sigma,\tau}$ is in 
${\cal T}_H$. For any triple $\sigma,\tau, \upsilon$ whose union is not the whole graph, we can define a set $\alpha$ by taking an element of each of $\sigma,\tau, \upsilon$ and a single element not in any of these. Then $\alpha \in I_E$, so $\sigma,\tau, \upsilon$ are $I_E$-adjacent. For any triple $\sigma,\tau, \upsilon$ whose union is the whole graph, we have
\[
T_{\sigma,\tau, \upsilon} = T_G + T_0 + T_\sigma + T_\tau + T_\upsilon \; .
\]
so this function is not independent of the type 1 and 2 functions. Since there is one function of type 1 and $2^{n-1}$ functions of type 2, the dimension of ${\cal T}_H^*$ in this case is at most $2^{n-1}+1$. On the other hand, the dimension of ${\cal T}_H$ is $2^n - |I_E| + 1 = 2^n -  \binom{n}{2} - \binom{n}{4}$. This is larger than $2^{n-1}+1$ for $n \ge 6$, so in this case ${\cal T}_H^*$ is smaller than ${\cal T}_H$.
\end{proof}
Of course, since ${\cal T}_H \subset {\cal T}_G = {\cal T}_G^*$, it is possible to build a basis of ${\cal T}_H$ using linear combinations of information basis functions. But the result implies that some of the basis elements will need to be linear combinations of information basis functions that are not individually in ${\cal T}_H$. 

\section{An equivalent quantum mechanics problem}

In this section, we note that there is an alternative characterization of functions in ${\cal T}_G$ as the space of ground state wavefunctions for a certain Hamiltonian for $n$ qubits. A similar characterization for ${\cal T}_H$ is possible, but we omit it here.

For a graph $G$ with $n$ vertices, define ${\cal H}_G$ to be the vector space of functions from subsets of vertices to $\mathbb{C}$ with inner product
\[
(T_1,T_2) = \sum_{\sigma} T_1(\sigma)^* T_2(\sigma) \; .
\]
This is the Hilbert space for a quantum mechanical system of $n$ qubits that we associate with the vertices of the graph.

We can identify an element $[\sigma]$ of the indicator basis (Definition \ref{def:indicator}) with an element $|\sigma \rangle$ of the standard $S_z$ eingenbasis where the $i$th spin is up iff $i \in \sigma$. That is 
\[
|\sigma \rangle \equiv \otimes_{i \in \sigma} |\uparrow \rangle \otimes_{i \notin \sigma} | \downarrow \rangle \; .
\]
We are interested in the subspace ${\cal H}^G_0$ for which the edge sums vanish. For an edge $e$ between vertices $i$ and $j$, we require that the sum of the coefficients of all $S_z$ basis elements with opposite spins for $i$ and $j$ vanish. This is the same as the vanishing of the inner product with the state (related to the states $\hat{\mu}_e$ above)
\[
|ij\rangle = (|\uparrow \rangle_i \otimes |\downarrow \rangle_j +  |\downarrow \rangle_i \otimes |\uparrow \rangle_j) \otimes_{k \ne i,j} (|\uparrow \rangle + |\downarrow \rangle )= P_{ij} |\Omega \rangle
\]
where $P_{ij}$ is a two-qubit gate on the $i$ and $j$ spins that projects on to the subspace generated by basis elements with opposite spins and $|\Omega \rangle$ is the state with every spin in the state $|\uparrow \rangle + |\downarrow \rangle$. Thus, our condition is
\[
\langle ij|\Psi \rangle = \langle \Omega | P_{ij} | \Psi \rangle = 0 \;  \qquad (\forall \{i,j\} \in E)
\]
So ${\cal H}_0$ is the orthogonal complement of the subspace generated by  the states $|ij \rangle$.

Defining another Hilbert space ${\cal H}_e$ with basis elements $|\bar{ij} \rangle$ associated with edges, we can say that the subspace we are interested in is the kernel of the map from ${\cal H}$ to ${\cal H}_e$ defined by the operator (closely related to ${\cal E}$ above)
\[
{\cal O}_G = \sum_{\{i,j\} \in E}|\bar{ij} \rangle \langle ij| \; .
\]
Equivalently, it is the kernel of the the positive operator ${\cal O}_G^\dagger {\cal O}_G$ that maps ${\cal H}$ to itself.\footnote{We thank Abhisek Sahu for this observation.} Thus, we can say that ${\cal H}_0$ is the space of ground states of the Hamiltonian
\[
H_G = {\cal O}_G^\dagger {\cal O}_G  =  \sum_{\{i,j\} \in E} |ij \rangle \langle ij | = \sum_{ij} P_{ij} P_{\Omega} P_{ij}  \; .
\]
An alternative form (up to normalization) is 
\[
H_G =  \sum_{\{i,j\} \in E} (I + X_i X_j + Y_i Y_j - Z_i Z_j) \otimes_{k \ne i,j} (I + X_k)
\]
where $X$, $Y$, and $Z$ are the standard Pauli operators and $I$ is the identity. The results of section 3 give us two different explicit ways to construct the space of ground states for this Hamiltonian.

\section{Algebraic definitions for divergence-free entropy sums}

The divergence-free entropy combinations that we have discussed are often defined by introducing a regulator (either in the theory or in the quantity itself) and then taking a limit where the regulator is removed. However, by using the language of algebraic quantum field theory, we can in some cases give independent definitions of the finite quantities that don't make reference to divergent entropies. See \cite{Haag1996, BratteliRobinson1997} for operator-algebraic preliminaries.
\begin{enumerate}
    \item The entropy of a region $A$ minus the entropy of its complement in a state $\rho$ can be alternatively written by choosing an arbitrary reference state $\sigma$ and defining 
    \[
S(A)-S(A^{c})
=\big\langle K^{\sigma}_{A}-K^{\sigma}_{A^{c}}\big\rangle_{\rho}
-\Big(S(\rho_A\Vert\sigma_A)-S(\rho_{A^{c}}\Vert\sigma_{A^{c}})\Big)\,
\]
Here, the first term is the expectation value of the full modular operator associated with the reference state and the second term is a difference of relative entropies, which are well-defined and finite even when we have a nontrivial boundary and the algebra is type III. Choosing $\sigma = \rho$, this reduces to
\[
S(A)-S(A^{c})
=\big\langle K^{\rho}_{A}-K^{\rho}_{A^{c}}\big\rangle_{\rho}
\]
the value assigned to the full modular Hamiltonian associated with region $A$ by the full state.
    \item 
    Assuming that the algebras ${\cal A}$ and ${\cal B}$ associated with regions $A$ and $B$ satisfy the split property (i.e. there exists a type I factor ${\cal N}$ such that ${\cal A} \subset {\cal N} \subset {\cal B}'$ where ${\cal B}'$ is the commutant of ${\cal B}$), there is an isomorphism between the algebra ${\cal A} \vee {\cal B}$ and the von Neumann tensor product algebra ${\cal A} \bar{\otimes} {\cal B}$. Thus, starting from a global state $\omega$ and the restrictions $\omega_{AB}$, $\omega_B$, and $\omega_A$ to the subalgebras ${\cal A} \vee {\cal B}$, ${\cal A}$, and ${\cal B}$, we can construct a state $\omega_A \bar{\otimes} \omega_B$ on ${\cal A} \vee {\cal B}$ that is the image of the product state $\omega_A \otimes \omega_B$ under the isomorphism. The mutual information $I(A,B)$ can be defined as the Araki relative entropy
    \[
    I(A,B) =S(\omega_{AB}|| \omega_A \bar{\otimes} \omega_B) \; 
    \]
    (see \cite{Araki1976} for relative entropy on von Neumann algebras and \cite{Fewster2016} for an exposition of the split property and its consequences).
    \item 
    With at least one non-adjacent region, we can express the tripartite information in terms of mutual informations as described in the introduction, and thus define it in terms of relative entropies as above. It seems reasonable to expect that the tripartite information (for cases where there are no higher-codimension divergences) can also be expressed somehow in terms of relative entropies in the general case, but we are not aware of an explicit construction that works when each region is adjacent to both of the other regions. This seems to be an interesting open question.
\end{enumerate}

\section*{Acknowledgements}

I would like to thank Jim Bryan, Pompey Leung, Abhisek Sahu, Jeremy van der Heijden, and Rana Zibakhsh for discussions.
I acknowledge support from the National Science and Engineering Research Council of Canada
(NSERC) and the Simons foundation via a Simons Investigator Award. 

\bibliographystyle{jhep}
\bibliography{refs}

\end{document}